\documentclass{sig-alternate-05-2015}
\usepackage{amsfonts,amssymb,amsmath,latexsym,ae,aecompl}
\usepackage{float}
\usepackage{epsfig}
\usepackage{enumerate}
\usepackage{graphicx}
\usepackage{fourier}
\usepackage{algorithmic}
\usepackage{caption}
\usepackage{graphicx}
\usepackage{subfig}
\usepackage{gensymb}
\usepackage{color}
\usepackage{booktabs}
\newtheorem{theorem}{Theorem}
\newtheorem{lemma}{Lemma}

\newtheorem{remark}{Remark}
\usepackage[ruled,vlined,linesnumbered]{algorithm2e}
\begin{document}

\setcopyright{acmcopyright}

\CopyrightYear{2017}
\setcopyright{acmcopyright}
\conferenceinfo{ICDCN '18,}{January 04-07, 2018, Varanasi, India}
\isbn{978-1-4503-4839-3/17/01}\acmPrice{\$15.00}
\doi{http://dx.doi.org/10.1145/3007748.3007786}

\title{Evacuating Two Robots from Two Unknown Exits on the Perimeter of a Disk}

\numberofauthors{4}
\author{
\alignauthor
Debasish Pattanayak\\
       \affaddr{IIT Guwahati, India}\\
       \email{p.debasish@iitg.ernet.in}
\alignauthor
H. Ramesh\\
 \affaddr{IIT Guwahati, India}\\       			
       \email{ramesh\_h@iitg.ernet.in}
\alignauthor
Partha Sarathi Mandal\\
        \affaddr{IIT Guwahati, India}\\
       \email{psm@iitg.ernet.in}
\and  
\alignauthor 
Stefan Schmid\\
      \affaddr{Aalborg University, Denmark and TU Berlin, Germany}\\
       \email{schmiste@gmail.com}
}

\maketitle
\begin{abstract}
Distributed evacuation of mobile robots is a recent development. 
We consider the evacuation problem of two robots which are initially located at the center of a unit disk. Both the robots have to evacuate the disk through the exits situated on the perimeter of the disk at an unknown location. 
The distance between two exits along the perimeter $d$ is given.
We consider two different communication models. First, in the wireless model, the robots can send a message to each other over a long distance.
Second, in face-to-face communication model, the robots can exchange information with each other only when they touch each other. 
The objective of the evacuation problem is to design an algorithm which minimizes the evacuation time of both the robots.  
For the wireless communication model, we propose a generic algorithm for two robots moving to two points on the perimeter with an initial separation of $\zeta \leq d$. We also investigate evacuation problem for both unlabeled and labeled exits in the wireless communication model.
For the face-to-face communication model, we propose two different algorithms for $\zeta =0$ and $\zeta =d$ for unlabeled exits. We also propose a generic algorithm for $\zeta \leq d$ for labeled exits. We provide lower bounds corresponding to different $d$ values in the face-to-face communication model. 
We evaluate the performance our algorithms with simulation for both of the communication models. 
\end{abstract}



\section{Introduction}
Searching is an inherent problem in computer science. Being an intriguing field with a long history, plethora of research had been conducted on search problems using a multitude of models. The models considered include probabilistic search model~\cite{stone1976theory}, cops and robbers model~\cite{bonato2011game}, search problem in group~\cite{ahlswede1987search}, classical pursuit and evasion~\cite{Chung2011,nahin2012chases} to name a few. In these papers, the goal is mainly to find an object located in a specific domain.

In this paper, the search problem requires the evacuation of robots from a disk. Unlike traditional search problems, we are trying to minimize the time needed for the \textit{last robot} to exit the disk. The search targets, i.e., the exits are located at arbitrary points on the perimeter not known a priori. The robots have no knowledge about the position of exits, but they have some information about the distribution of exits. 
The time for evacuation can be considered to be the competitive ratio of the algorithm, given the robots have a speed of 1 unit distance per unit time. Here the competitive ratio is the ratio of maximum distance traveled by any robot over all possible position of exits on the perimeter and the minimum distance between the starting position of the robots and location of exit.
The robots cooperate with each other to locate the exits. If there are two robots, then we need both the robots to exit the disk as soon as possible and the competitive ratio is the supremum of the maximum distance traveled by any of the two robots.

We study the evacuation from a disk for two robots and two exits. The robots start from the center of the disk and attempt to reach exits located at some point on the perimeter. The robots can move anywhere within the perimeter of the disk. 
Each robot has the knowledge of the distribution of exits over the perimeter of the unit disk. 
The robots can communicate with each other in a \textit{wireless} manner or when they come in contact, which is termed as \textit{face-to-face communication}.
We introduce the notion of \textit{labeled} exits, where the exits have identities and  relative ordering.
\subsection{Our Contribution} 
In this paper, we propose following evacuation algorithms for two different communication models where two robots evacuate via the exits situated on the perimeter of a unit disk starting from the center. We consider both unlabeled and labeled exits.
\begin{itemize}
\item In wireless communication model, we propose a generic algorithm for evacuation of robots moving to the perimeter with separation $\zeta \leq d$. 

\item In the face-to-face communication model, we use the meeting of the robots to gain more information about search space and achieve an agreement, while reducing the time required to evacuate and propose solutions to the problem for $\zeta = 0$ and $\zeta = d$.
\item We propose lower bounds for specific values of $d$ in the face-to-face communication model. 
\item We compare the algorithms with simulations in both communication models for both labeled and unlabeled exits which provides insights about how the worst-case evacuation time varies with $d$.
\end{itemize}
 
\subsection{Related Works}
 Czyzowicz et al.~\cite{CzyzowiczGGKMP14} introduced the search algorithm for robots in a disk. In \cite{CzyzowiczGGKMP14}, they proposed two different communication models namely, \textit{wireless} and \textit{face-to-face}. In the wireless model, robots can communicate at any point in time. So the robot that finds the exit can immediately communicate the location of exit and then the other robot can also exit the disk. They provided optimal worst-case evacuation time for two robots in wireless model, as $1+\frac{2\pi}{3}+\sqrt{3}\approx 4.826$. 
In the face-to-face communication model, \cite{CzyzowiczGGKMP14} provided upper bound of 5.740 and lower bound 5.199 for the single exit with two robots. After that Czyzowicz et al.~\cite{CzyzowiczGKNOV15}, again improved the lower bound and upper bound for two robots with one exit in a disk, by proposing linear and triangular detours at the worst case positions. They improved the upper bound to 5.628 and lower bound to $3+\frac{\pi}{6}+\sqrt{3}\approx 5.255$. 
In \cite{CzyzowiczDGKM16}, the search domain is restricted to only the perimeter of the circle. The robots have wireless communication. They provide upper and lower bounds for different distribution of exits.
In \cite{LamprouMS16}, the robots with different speeds are considered in wireless communication model. 
Borowiecki et al.~\cite{Borowiecki0DK16} have explored evacuation from graphs, where some of the vertices in a graph contain exits. The robots in the distributed model try to exit from the graph, where they change their strategy in an online fashion while they obtain more information as they move.
The problem considered in this paper can be viewed as a generalized version of the gathering problem \cite{FlocchiniPSW05} with multiple predefined gathering locations.

The rest of the paper is organized as follows. First, Section~\ref{sec:prelim} introduces notations and conventions to be used. 
In Section~\ref{sec:wlevac}, we propose a generic algorithm for evacuation in the wireless communication model, which works for any initial separation of robots less than $d$ for both unlabeled and labeled exits.
In Section~\ref{sec:f2f}, we propose two evacuation algorithms in the face-to-face communication model  for $\zeta = 0$ and  $\zeta = d$ for unlabeled exits. Also we propose a generic algorithm for labeled exits and provide lower bounds for different values of $d$. Further in Section~\ref{sec:simumation}, we provide simulation results for our algorithms and compare them. Finally we conclude in Section~\ref{sec:conclusion}.

\section{Preliminaries} \label{sec:prelim}
The disk considered is a unit disk.
The two robots are denoted as $R_1$ and $R_2$.
$d$ is the length of the smaller arc between two exits located on the perimeter.
The length of a chord corresponding to an arc of length $d$ is $ 2\sin (d/2)$.
$\widearc{AB}$ denotes the smaller arc over the perimeter joining two points $A$ and $B$. In general, $\widearc{AB}$ denotes the arc starting from $A$ and moving in the clockwise direction until $B$. So $\widearc{AB}$ and $\widearc{BA}$ combined make the circle. $\overline{AB}$ denotes the line segment joining points $A$ and $B$. We also use $\widearc{AB}$ and $\overline{AB}$ to denote the length of arc and line segment respectively. In this paper, we compute the evacuation time starting from the point on perimeter. So the total evacuation time is the computed evacuation time added by 1, which is the time required to reach the point on the perimeter from $O$. \vspace{2em}

\subsection{Model}
The identical point robots can move anywhere within the disk. 
The two robots travel with uniform speed, i.e., 1 unit distance per unit time.
The robots can have two different modes of communication, wireless or face-to-face. 
They have the ability to solve trigonometric equation.

\section{Evacuation in Wireless Communication Model}\label{sec:wlevac}
This section considers robots with wireless communication capability. The robots can send and receive messages in a wireless manner even if they are present at different positions at a particular time instance. 
In this section we present Algorithm~\ref{algo:evacuateWL} (\textsc{EvacuateWL}) for evacuation. The robots start moving from the center towards the boundary and start searching for the exits. The algorithm is evoked if a robot encounters an exit or receives a message. The message exchanged by the robots can just be a one bit message. Messages are reliable and  message propagation delay is ignored.

\subsection{Algorithm for Evacuation}
The two robots $R_1$ and $R_2$ start at the center of the circle $O$. They move to two points $B$ and $C$ on the perimeter. We consider $O$ as the origin and $\overrightarrow{OA}$ as the positive $x$-axis. 
Say $\widearc{BC} = \zeta$. 
 For simplicity, consider $A$ as the midpoint of $\widearc{BC}$. 
 So $\widearc{BA}=\widearc{AC}=\zeta/2$.
 Say, $R_1$ moves in the counter-clockwise direction and $R_2$ in the clockwise direction. Suppose two exits are located at $E_1$ and $E_2$.
Without loss of generality, let us consider that $R_1$ finds the exit $E_1$ at $X$ before $R_2$ unless both find the exits simultaneously.
Now, $R_1$ sends a message to $R_2$ that it found the exit. Given the two robots have the same velocity, $R_2$ can find out the location of the exit $E_1$ based on the position of $R_1$.
Say $\widearc{XB} = \widearc{CD} = x$. Since the robots know $d$, the distance over arc between the exits, they can predict two probable positions for the exit at $E_1'$ and $E_2'$.
Let $E_1'$ be the closest exit in clockwise direction and $E_2'$ is the closest exit in counter-clockwise direction from $X$.
There can be four different cases based on the value of $x$ and $\zeta$.
\begin{algorithm}
\caption{\textsc{EvacuateWL}}\label{algo:evacuateWL}
\SetKwInOut{Input}{Input}\SetKwInOut{Output}{Output}
\Input{Center and radius of the circle, and distance $d$ between the two exits}
\Output{Path of evacuation}
\If{encounters an exit}
	{Send a message to other robot\\
	Evacuate through that exit}
\If{encounters an exit and receives a message simultaneously}
	{Evacuate through that exit}
\If{receives a message}
	{Current position of the robot is $D$\\
	The position of other robot is $X$\\
	Determine probable exit positions $E_1'$ and $E_2'$\\
	Check if the locations $E_1'$ and $E_2'$ are explored or not\\
	\If{none of them are explored}
		{The path is
		$\min(\overline{DX},\overline{DE_1'} + \overline{E_1'E_2'},\overline{DE_2'} + \overline{E_1'E_2'})$}
	\If{one of them is explored}
		{Say $E_1'$ is already explored\\
		The path is $\min(\overline{DX},\overline{DE_2'})$}
	Move along the path to find an exit\\
	Evacuate through that exit
	}

\end{algorithm}

Algorithm~\ref{algo:evacuateWL} can be described as following.
\begin{description}
\item[\textbf{Case 1:}] Both $E_1'$ and $E_2'$ are unexplored\\
	In this case, $R_2$ can move towards the definite exit $X$ or it can go to the two probable exit positions $E_1'$ and $E_2'$.
	It chooses the minimum of the two.
	There can be three different situations in this case.
	\begin{itemize}
	\item $ 2x+\zeta \leq d$\\	
	The time for evacuation is $\widearc{CD} + \min(\overline{DX}, \overline{DE_1'} + \overline{E_1'E_2'})$, i.e., $ x + \min(2\sin(x + \zeta/2), 2\sin(d/2 - x - \zeta/2) + 2 \sin(\pi - d))$ as shown in Fig.~\ref{fig:wldiffc11}.

	\begin{figure}[H]
	\centering
	\includegraphics[height=0.3\linewidth]{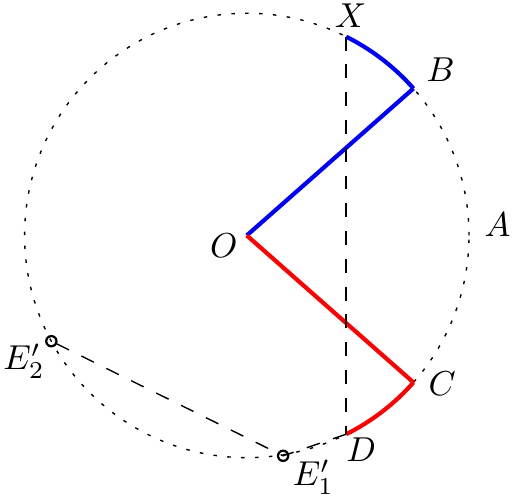}
	\caption{Both $E_1'$ and $E_2'$ are in $\widearc{DX}$}\label{fig:wldiffc11}
	\end{figure}\vspace{-2em}
	
	\item 	$ x \leq d < x + \zeta $ and $d + 2x + \zeta < 2\pi$\\
	The time for evacuation is $\widearc{CD} + \min( \overline{DX}, \overline{DE_1'} + \overline{E_1'E_2'},$ $ \overline{DE_2'} + \overline{E_1'E_2'})$, i.e., $x + \min(2\sin(x+\zeta/2), 2\sin(x + \zeta/2 - d/2) + 2\sin(\pi - d), 2\sin(x + \zeta/2 + d/2)+2\sin(\pi - d))$ as shown in Fig.~\ref{fig:wldiffc12}.
	
	\begin{figure}[H]
	\centering
	\includegraphics[height=0.3\linewidth]{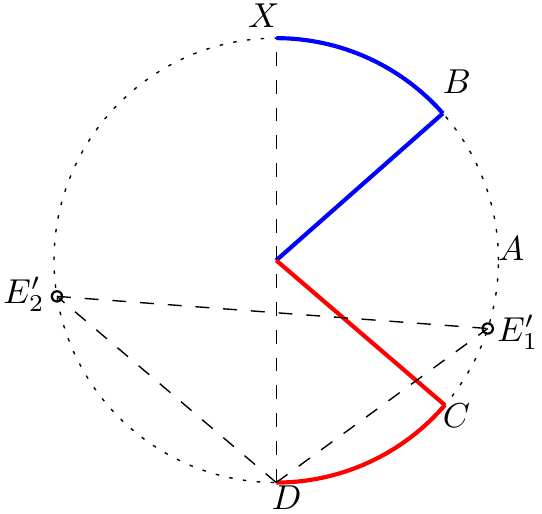}
	\caption{$E_1'$ is in $\widearc{DX}$ and $E_2'$ is in $\widearc{BC}$}\label{fig:wldiffc12}
	\end{figure}\vspace{-2em}
	\item 	$ x \leq d < x + \zeta $ and $d + x + \zeta > 2\pi$\\
	The time for evacuation is $\widearc{CD} + \min(\overline{DX}, \overline{DE_2'} + \overline{E_1'E_2'})$, i.e, $ x + \min(2\sin(x + \zeta/2), 2\sin(d/2 + x + \zeta/2 - \pi) + 2\sin(\pi - d)) $ as shown in Fig.~\ref{fig:wldiffc13}.
	\begin{figure}[H]
	\centering
	\includegraphics[height=0.3\linewidth]{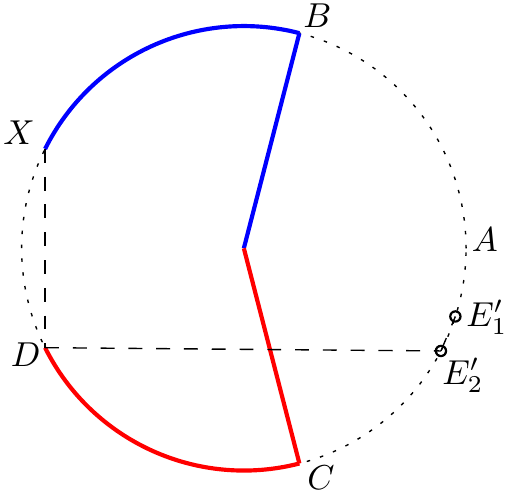}
	\caption{$E_1'$ and $E_2'$ are in $\widearc{BC}$}\label{fig:wldiffc13}
	\end{figure}\vspace{-2em}
	\end{itemize}
\item[\textbf{Case 2:}] $E_1'$ is unexplored, i.e., $x < d$ and $ x + \zeta +d < 2\pi <  2x + \zeta + d$\\
	In this case, as $E_2'$ is already explored, so there is definitely an exit at $E_1'$. So the time for evacuation is $\widearc{CD} + \min(\overline{DX}, \overline{DE_1'})$, i.e., $x + \min(2\sin(x + \zeta/2), 2\sin(x + \zeta/2 - d/2))$  as shown in Fig.~\ref{fig:wldiffc21}.
	\begin{figure}[H]
	\centering
	\includegraphics[height=0.3\linewidth]{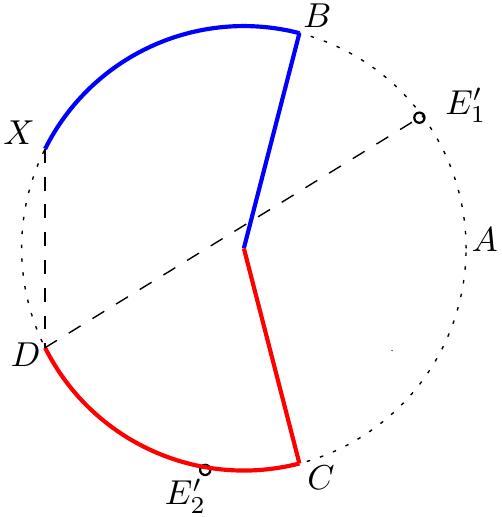}
	\caption{$E_1'$ is in $\widearc{BC}$ and $E_2'$ is in $\widearc{CD}$}\label{fig:wldiffc21}
	\end{figure}\vspace{-2em}
\item[\textbf{Case 3:}] $E_2'$ is unexplored\\
	In this case, as $E_1'$ is already explored, so there is definitely an exit at $E_2'$. So the time for evacuation is $\widearc{CD} + \min(\overline{DX}, \overline{DE_2'})$, i.e., $x + \min(2\sin(x + \zeta/2), 2\sin(x + \zeta/2 + d/2))$.
	 This case happens in the following two situations.
	\begin{itemize}
	\item $ x+\zeta \leq d < 2x + \zeta$ and $d + 2x + \zeta < 2\pi$\\ 
	$E_1'$ is on $\widearc{CD}$ and $E_2'$ lies on $\widearc{DX}$ as shown in Fig.~\ref{fig:wldiffc31}.
	\begin{figure}[H]
	\centering
	\includegraphics[height=0.3\linewidth]{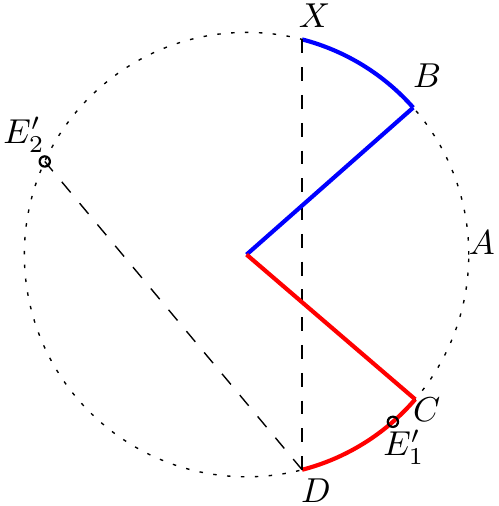}
	\caption{$E_1'$ is on $\widearc{CD}$ and $E_2'$ lies on $\widearc{DX}$}\label{fig:wldiffc31}
	\end{figure}\vspace{-2em}

	\item $ d < x$ and and $d + 2x + \zeta < 2\pi$\\
	$E_1'$ is on $\widearc{XB}$ and $E_2'$ lies on $\widearc{DX}$ as shown in Fig.~\ref{fig:wldiffc32}.
	\begin{figure}[H]
	\centering
	\includegraphics[height=0.3\linewidth]{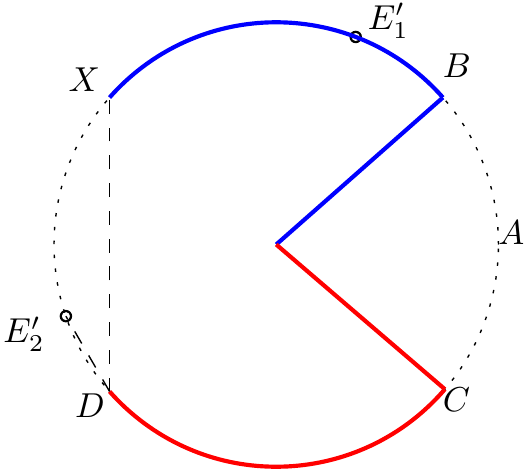}
	\caption{$E_1'$ is on $\widearc{XB}$ and $E_2'$ lies on $\widearc{DX}$}\label{fig:wldiffc32}
	\end{figure}\vspace{-2em}
	\end{itemize}
\end{description}

\begin{theorem}
If at least a robot encounters an exit, then an agreement for evacuation path in the wireless communication model is achieved when the other robot receives the message.
\end{theorem}
\begin{proof}
If both the robots find the exits simultaneously, then they exchange messages and an agreement is achieved. Both the robots exit via their respective exits. 

If one robot finds the exit, then it sends a message to other robot. When the second robots receives the message, it can determine the location of one exit and two probable exit locations. 
Then it determines the path of evacuation according to Algorithm~\ref{algo:evacuateWL}. Hence an agreement is achieved.
\end{proof}

The cases mentioned above are applicable for cases where $0 \leq \zeta \leq d$.
If $ d < \zeta$, then the worst case arises in the situation where both the exits lie on the unexplored part of the perimeter as shown in Fig.~\ref{fig:wldiffc5}. So, the robots have to go back the unexplored part of the perimeter, when they do not find the exits. This increases the evacuation time compared to Algorithm~\ref{algo:evacuateWL} for $\zeta \leq d$. This case can be considered independent of the series of cases described before this, since the condition for this case is independent of $x$. In this special case the time of evacuation is always greater than $\widearc{DB} + \overline{DB}$, i.e., $\pi - \zeta/2 + 2\sin(\pi - \zeta/2)$. 
\begin{figure}[H]
\centering
\includegraphics[height=0.3\linewidth]{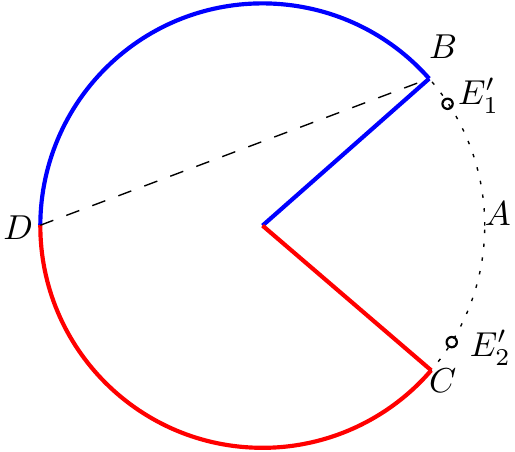}
\caption{$\zeta > d$}\label{fig:wldiffc5}
\end{figure}\vspace{-2em}

\begin{remark} It is easy to observe that if both the robots travel in the same direction (say clockwise) starting from different points, then the time for evacuation increases compared to Algorithm~\ref{algo:evacuateWL} in the worst case.
\end{remark}

\subsection{Labeled Exits}
Consider the exits are labeled. This means when a robot encounters an exit, it can identify whether the exit is $E_1$ or $E_2$. The robots have prior knowledge that the $E_2$ lies at a distance $d$ from $E_1$ in the counter-clockwise direction. 
This provides the robots with enough information to locate both the exits on the perimeter. If $R_1$ finds the exit $E_1$ first, then $R_2$ knows the exact locations of both exits and it can move towards the nearest exit.

There can be the following two situations if $E_1$ is located as $X$.
\begin{description}
\item[\textbf{Case 1:}] $E_2$ lies on $\widearc{DX}$\\
The time for evacuation is $\widearc{CD} + \min(\overline{DX}, \overline{DE_2})$, i.e., $x + \min(2\sin(x + \zeta/2), 2\sin(x + \zeta/2 + d/2))$ as shown in Fig.~\ref{fig:wllower1}.
	\begin{figure}[H]
	\centering
	\includegraphics[height=0.3\linewidth]{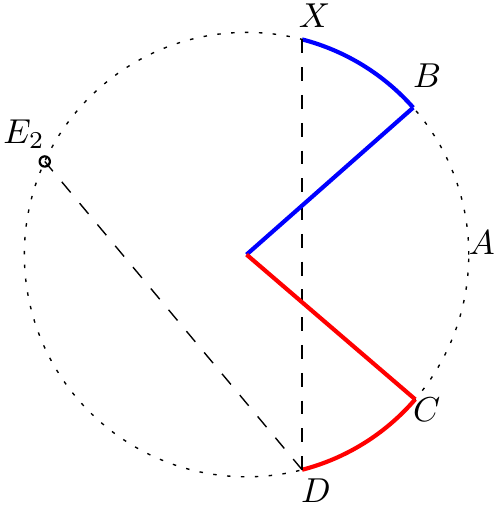}
	\caption{$E_2$ lies on $\widearc{DX}$}\label{fig:wllower1}
	\end{figure}\vspace{-2em}
\item[\textbf{Case 2:}] $E_2$ lies on $\widearc{BC}$\\
The time for evacuation is $\widearc{CD} + \min(\overline{DX}, \overline{DE_2})$, i.e., $x + \min(2\sin(x + \zeta/2), 2\sin(d/2 + x + \zeta/2 - \pi))$ as shown in Fig~\ref{fig:wllower2}.
	\begin{figure}[H]
	\centering
	\includegraphics[height=0.3\linewidth]{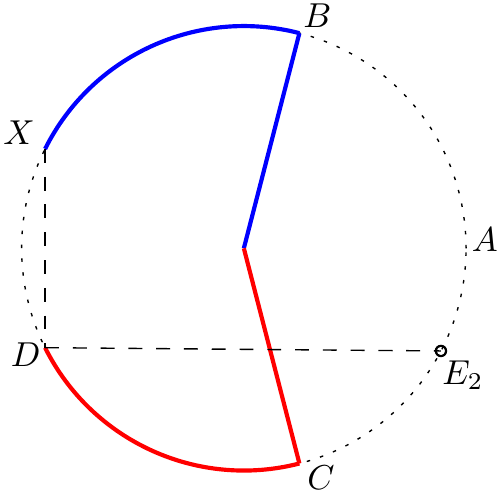}
	\caption{$E_2$ lies on $\widearc{BC}$}\label{fig:wllower2}
	\end{figure}\vspace{-2em}
\end{description}

Similarly, we can consider $E_2$ is encountered at $X$ by $R_1$. The analogous cases would be $E_1$ either lies on $\widearc{DX}$ or $\widearc{BC}$.

\section{Evacuation in Face-to-Face Communication Model}\label{sec:f2f}
This section presents algorithms for robot evacuation with face-to-face communication, i.e., the robots can exchange information only when they lie on the same point simultaneously. In this model, we consider only two specific cases for the value of $\zeta$. Those are $\zeta = 0$, i.e., both the robots moving to the same point on the perimeter and $\zeta = d$, i.e., both the robots moving to different points on the perimeter which are $d$ distance apart from each other.
\subsection{Same Point on Perimeter ($\zeta =0$)}\label{sec:samepoint}

In this section, we describe an evacuation algorithm where the robots move from the center of the disk to the same point on the perimeter. Initially, $R_1$ and $R_2$ move together from center of the disk, $O$, to an arbitrary point $A$ on the perimeter. Suppose two exits are located at $E_1$ and $E_2$. The robots $R_1$ and $R_2$ move in the different direction starting from $A$.
We consider $O$ as the origin and $\overrightarrow{OA}$ as the positive $x$-axis.

Suppose $R_1$ encounters exit $E_1$ at $X$, where $\widearc{XA}=x$.
$R_1$ computes two probable exit positions $E_1'$ and $E_2'$ such that $\widearc{E_2'X} =d$ and $\widearc{XE_1'} = d$ as shown in Fig.~\ref{fig:samepointcase1b}. 
 $R_1$ then computes a point $M$ such that $\widearc{XA}+\overline{XM}=\widearc{AM}$. Say $\widearc{AM}=y$. Consequently, $\overline{XM} = 2\sin((x+y)/2)$. Now, $y$ is the solution to the following equation,
\begin{equation}\label{eq:yvalue}
 x+ 2 \sin\left(\frac{x+y}{2}\right) = y
\end{equation}
Based on the relations of $x$, $y$ and $d$, we have the following four cases. 

\begin{description}
\item[Case 1: ] $x+y\leq d$\\
In this case,  $R_1$ catches $R_2$ before it can encounter closest probable exit $E_1'$.
As shown in Fig.~\ref{fig:samepointcase1b}, if both robots return to $X$ for exiting the disk, then the time for evacuation is $\widearc{AM}+\overline{XM} = 2y-x$.
Instead of returning to $X$, the robots can go to $E_1'$ and $E_2'$ since one of them contains an exit. If $E_1'$ is the other exit, then evacuation time is $y+\overline{ME_1'}$, otherwise the evacuation time is $y+\overline{ME_1'}+\overline{E_1'E_2'}$. So the time for evacuation in the worst case is $y + \min(y-x, \overline{ME_1'} + \overline{E_1'E_2'})$.

\begin{figure}[H]
\centering
\includegraphics[height=0.3\linewidth]{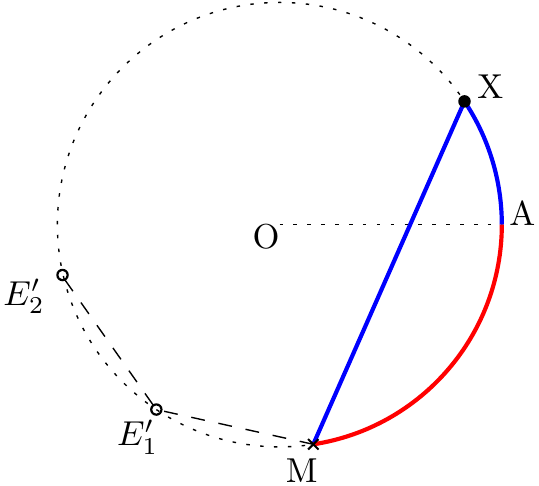}
\caption{$R_1$ catches $R_2$ before probable exit and goes to either the definite exit $X$ or the probable exits $E_1'$ and $E_2'$}\label{fig:samepointcase1b}
\end{figure}\vspace{-2em}

\item[Case 2:] $d < x + y $ and $x \leq d/2$\\
	\textit{Case 2a: }$R_1$ tries to catch $R_2$ on the circle if $y \leq \widearc{AE_2'}$. If it finds $R_2$ on the circle, then there is no exit present at $E_1'$. So the time for evacuation is $y +\min(\overline{MX}, \overline{ME_2'})$ if $R_1$ catches $R_2$ as shown in Fig.~\ref{fig:samepointcase2a}. If there is an exit present at $E_1'$, then $R_2$ would catch $R_1$ on $\overline{XM}$ according to \textit{Case 3a}.
	
\begin{figure}[H]
\centering
\includegraphics[height=0.25\linewidth]{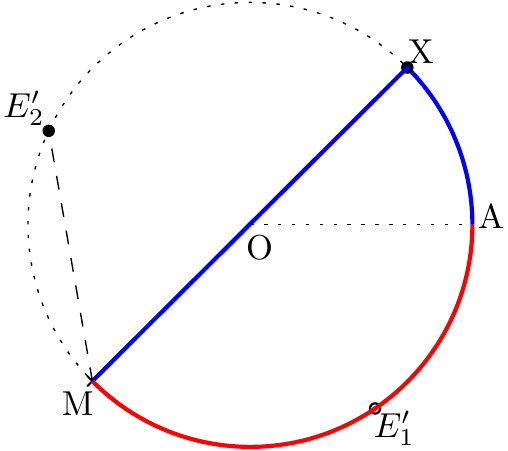}
\caption{$R_1$ tries to catch $R_2$ at $M$ and returns to closest exit.}\label{fig:samepointcase2a}
\end{figure}\vspace{-2em}
	\textit{Case 2b: } If $\widearc{AM} > \widearc{AE_2'}$, $R_1$ exits at $X$ as shown in Fig.~\ref{fig:samepointcase2b}. The time for evacuation in the worst case is $2\pi -d -x$.
\begin{figure}[H]
\centering
\includegraphics[height=0.3\linewidth]{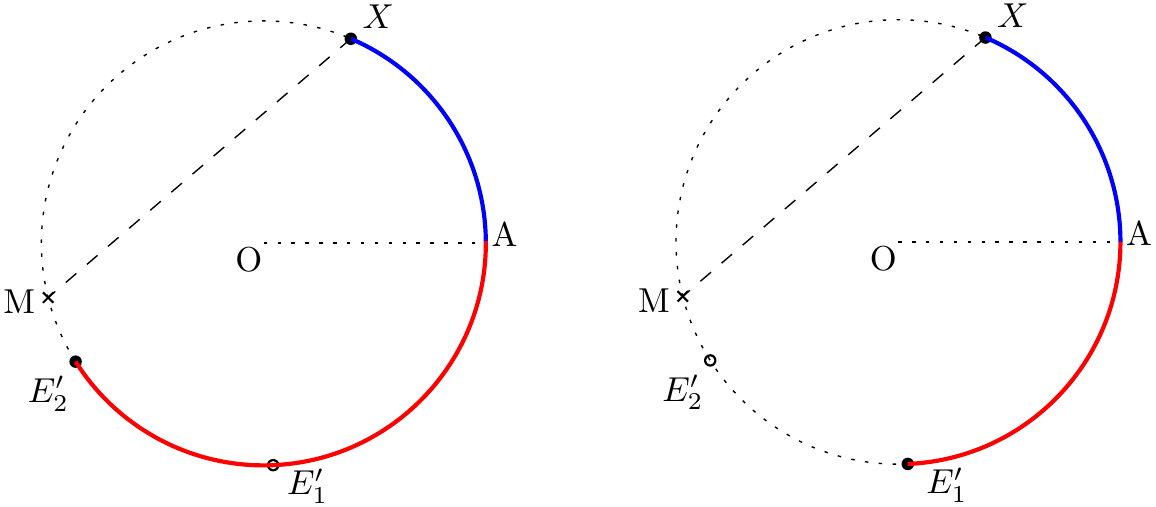}
\caption{$R_1$ exits at $X$ since $M$ is beyond $E_2'$.}\label{fig:samepointcase2b}
\end{figure}\vspace{-2em}	

\item[Case 3:] $ d/2 < x < d$\\
	In this case, when $R_1$ encounters an exit at $X$, $R_2$ may have already encountered another exit at $E_1'$. $R_1$ computes $E_3'$ such that $\widearc{E_1'E_3'} = d$ and $M'$ such that $\widearc{AE_1'} + \overline{E_1'M'} = \widearc{M'A}$. 
	\textit{Case 3a: } If $\widearc{M'A} < \widearc{E_3'A}$, then $R_1$ tries to catch $R_2$ on $\overline{E_1'M'}$, if $R_2$ is trying to catch $R_1$ on the circle (i.e., $R_2$ satisfies \textit{Case 2a}). In this case $R_1$ computes point $N$ on $\overline{E_1'M'}$ such that $\widearc{XA} + \overline{XN} = \widearc{AE_1'} + \overline{E_1'N}$. 
\begin{itemize}
	
\item If $R_1$ finds $R_2$ at $N$, then $R_1$ and $R_2$ have encountered exits at $X$ and $E_1'$ respectively. So the time for evacuation is $\widearc{XA} + \overline{XN} + \min(\overline{NX}, \overline{NE_1'})$ (ref. Fig.~\ref{fig:samepointcase3a}).

\begin{figure}[H]
\centering
\includegraphics[height=0.3\linewidth]{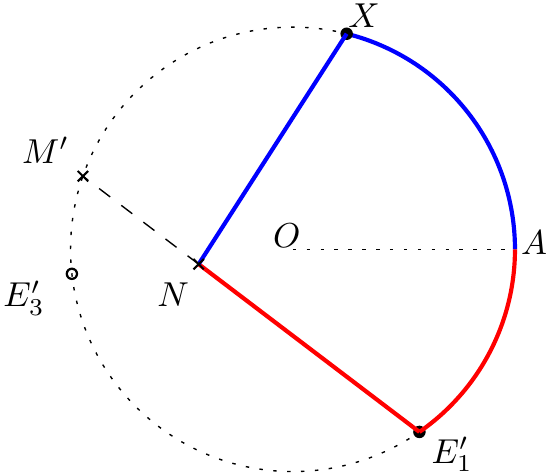}
\caption{$R_1$ tries to catch at $N$}\label{fig:samepointcase3a}
\end{figure}\vspace{-2em}	

\item If $R_1$ does not find $R_2$ at $N$ implies that $R_2$ did not encounter an exit at $E_1'$. So $R_1$ finds point $P$ such that $\widearc{AP} = \widearc{XA} + \overline{XN} + \overline{NP}$. 
\begin{itemize}

\item If $\widearc{AP} <\widearc{AE_2'}$, then $R_1$ catches $R_2$ at $P$. The time for evacuation is $\widearc{AP} + \min(\overline{PE_2'}, \overline{PX})$ as per Fig.~\ref{fig:samepointcase3ca}.

\begin{figure}[H]
\centering
\includegraphics[height=0.3\linewidth]{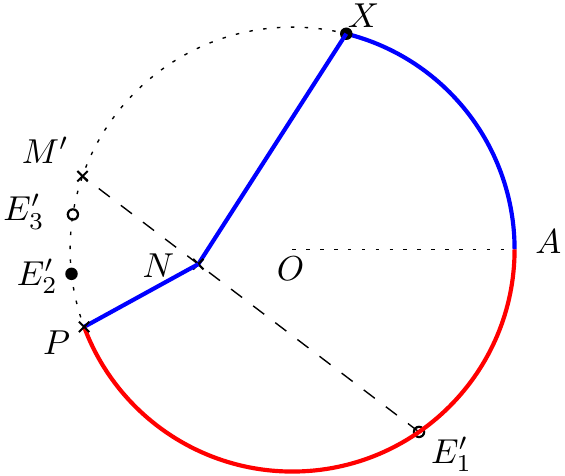}
\caption{$R_1$ catches $R_2$ at $P$ beofre $E_3'$.}\label{fig:samepointcase3ca}
\end{figure}\vspace{-1em}

\item If $\widearc{AP} \geq \widearc{AE_2'}$, $R_1$ goes to the closest exit from $N$ as shown in Fig.~\ref{fig:samepointcase3cb}. The time for evacuation is $\widearc{XA}+ \overline{XN} + \min(\overline{NX}, \overline{NE_2'})$
\end{itemize}

\begin{figure}[H]
\centering
\includegraphics[height=0.3\linewidth]{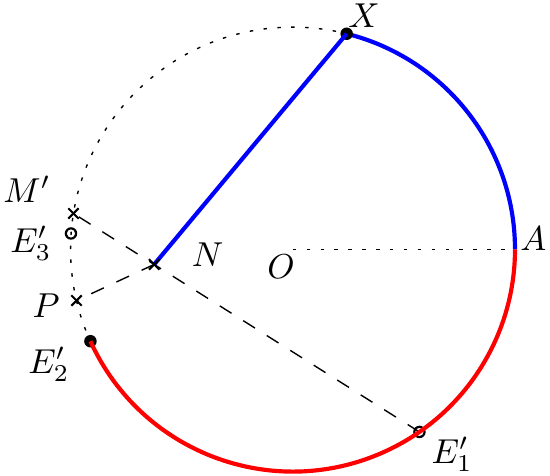}
\caption{$R_1$ moves to closest exit from $P$}\label{fig:samepointcase3cb}
\end{figure}\vspace{-2em}	

\end{itemize}
\textit{Case 3b: } If $\widearc{M'A} \geq \widearc{E_3'A}$, $R_1$ exits at $X$.  The time for evacuation in the worst case is $\max(x, 2\pi - x - d)$ as per Fig.~\ref{fig:samepointcase3b}.
\begin{figure}[H]
\centering
\includegraphics[height=0.3\linewidth]{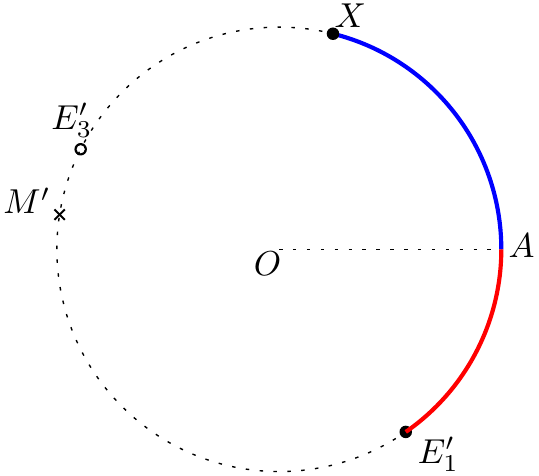}
\caption{$R_1$ exits at $X$ since $M'$ is beyond $E_3'$.}\label{fig:samepointcase3b}
\end{figure}\vspace{-2em}	

%

\item[Case 4:] $d\leq x$\\
When $R_1$ encounters the exit at $X$, it already knows that $E_2'$ is the other exit. $R_1$ computes $M$ such that $\widearc{AM} = \widearc{XA} +\overline{XM}$.

\textit{Case 4a: }If $\widearc{AM}<\widearc{AE_2'}$, then $R_1$ catches $R_2$ before $R_2$ encounters an exit as shown in Fig.~\ref{fig:samepointcase4a}. So the evacuation time is $\widearc{AM}+\min(\overline{ME_2'},\overline{MX})$.
\begin{figure}[H]
\centering
\includegraphics[height=0.25\linewidth]{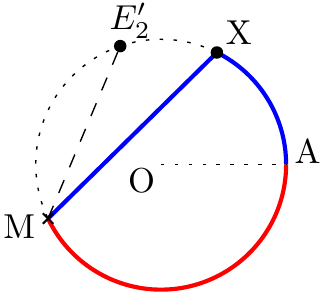}
\caption{$R_1$ catches $R_2$ at $N$ after missing at $M$}\label{fig:samepointcase4a}
\end{figure}\vspace{-1em}
\textit{Case 4b: }If $\widearc{AM}\geq \widearc{AE_2'}$, then $R_2$ reaches exit at $E_2'$ before $R_1$ can catch it as shown in Fig.~\ref{fig:samepointcase4b}. So $R_1$ and $R_2$ exit from $X$ and $E_2'$ respectively. So the time of evacuation is $\max(\widearc{XA},\widearc{AE_2'})$, i.e., $\max(x,2\pi-d-x)$.
\begin{figure}[H]
\centering
\includegraphics[height=0.25\linewidth]{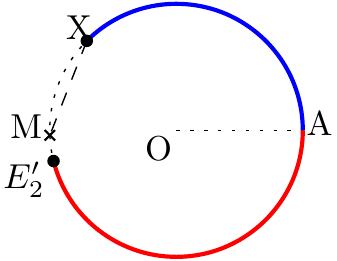}
\caption{$R_2$ reaches exit before $M$}\label{fig:samepointcase4b}
\end{figure}\vspace{-1em}
\textit{Case 4c: }If $2x + d \geq 2\pi$, then $R_2$ has already encountered exit by the time $R_1$ has encountered an exit. So $R_1$ exits at $X$. The time for evacuation is $x$ as per Fig.~\ref{fig:samepointcase4c}.
\begin{figure}[H]
\centering
\includegraphics[height=0.25\linewidth]{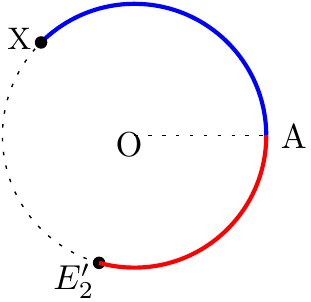}
\caption{$R_2$ reaches exit before $R_1$ reaches $X$}\label{fig:samepointcase4c}
\end{figure}\vspace{-2em}

\end{description}

\begin{theorem}\label{thm:f2fzeta0}
An agreement for evacuation is achieved for $\zeta = 0$ in the face-to-face communication model.
\end{theorem}
\begin{proof}

There are broadly two different situations for two robots to have an agreement for evacuation. In most of the cases, the robots meet with each other, exchange information and evacuate the disk via the nearest exit. But in some other cases, the robots agree on their respective path of evacuation without meeting each other, since they can obtain the complete information about the behavior of the other robot at some intermediate point. The algorithm for $\zeta = 0$ in face-to-face communication model guarantees that the behavior of the two robots are consistent in all the situations, i.e., either both meet or both exit separately without meeting.
The cases  are as following.

\textbf{Case 1:} $x + y \leq d$, in this case, an agreement is achieved at $M$, when $R_1$ catches $R_2$. $R_2$  was following the circular path till $M$ before it got caught (ref. Fig.~\ref{fig:samepointcase1b}).

\textbf{Case 2:} $ d< x + y$ and $x \leq d/2$, in this case, an agreement can be achieved at $M$, if $R_1$ catches $R_2$ (ref. Fig.~\ref{fig:samepointcase2a}). Here $R_2$ follows Case 3. So $R_1$ can get caught on $\overline{XM}$. If $R_1$ gets caught before it reaches $M$, then it obtains the locations of both exit and an agreement is achieved (ref. Fig.~\ref{fig:samepointcase3a}). 

If $x + y > 2\pi - d$, then $R_1$ exits at $X$, since $R_1$ realizes that it cannot catch $R_2$ before it encounters an exit. If $R_2$ encounters an exit, $R_2$ also evacuates at that exit (ref. Fig~\ref{fig:samepointcase3b}).

\textbf{Case 3:} $d/2 < x < d$, in this case, an agreement is achieved at either $N$ (ref. Fig~\ref{fig:samepointcase3a}) or $X$ (ref. Fig.~\ref{fig:samepointcase3b}).
For an agreement at $N$, the robot $R_2$ was following Case 2 (ref. Fig.~\ref{fig:samepointcase2a}) at the time of meeting. 
$R_1$ can always predict the exact path of $R_2$ once it has the complete information regarding exits. In other words, complete information leads to an agreement. Here $R_1$ can obtain complete information at $N$, i.e., it knows the location of other exit and the behavior of $R_2$ (ref. Fig.~\ref{fig:samepointcase4b}). Then both $R_1$ and $R_2$ can achieve an agreement together at $P$ (ref. Fig.~\ref{fig:samepointcase3ca}) or separately at $N$ and $E_2'$ (ref. Fig.~\ref{fig:samepointcase3cb}).

$R_1$ can also achieve an agreement at $X$, if $R_2$ has already exited at $E_1'$ according to Case 2 (ref. Fig.~\ref{fig:samepointcase2b}). If the other exit is at $E_2'$, then $R_2$ follows either Case 3 (ref. Fig.~\ref{fig:samepointcase3b}) or Case 4 (ref. Fig.~\ref{fig:samepointcase4b}) and exits at $E_2'$.

\textbf{Case 4:} $ d \leq x$, in this case, $R_1$ already knows the position of other exit and can predict the behavior of $R_2$. So it obtains the information at $X$, and achieves an agreement at $X$ (ref. Fig.~\ref{fig:samepointcase4b} and \ref{fig:samepointcase4c}) or at $M$ (ref. Fig.~\ref{fig:samepointcase4a}). Here $R_2$ can follow Case 3 (ref. Fig.~\ref{fig:samepointcase3b}) or Case 4 (ref. Fig.~\ref{fig:samepointcase4b}).
\end{proof}

\subsection{Different Points on Perimeter ($\zeta = d$)}\label{sec:diferentpoint}
Consider $R_1$ and $R_2$ start from points $B$ and $C$ such that $\widearc{BC}=d$. 
They start moving in opposite directions. If the initial separation is more than $d$, then there is a chance that both the exit lie in that arc and remain unexplored. In this section, we can reduce the search space by $d$, since at least one exit exists in the remaining perimeter.
Similar to Section~\ref{sec:samepoint}, we consider $O$ as the origin and $\overrightarrow{OA}$ as the positive $x$ axis. Taking $A$ as the midpoint of $\widearc{BC}$ gives us $\widearc{BA}=\widearc{AC}=d/2$.

Without loss of generality, consider $R_1$ finds exit $E_1$ at $X$, where $\widearc{XB}=x$. $R_1$ computes the probable exit position $E_2'$ at a distance $d$ over the perimeter from $X$, such that $\widearc{E_2'X}=d$. 
There are two cases with respect to the value of $x$ and $d$.
\begin{description}
\item[Case 1:] $x < d$\\

$R_1$ computes the position of $M$ such that $\widearc{XB}+\overline{XM}=\widearc{CM}$. Say $\widearc{CM}=y$. Now, $y$ is a solution to the following equation
\begin{equation}
y=x+2\sin\left(\cfrac{x+y+d}{2}\right)
\end{equation}

\textit{Case 1a:}  If $E_2'$ is not located in $\widearc{CM}$, then the robots have to return to the known exit at $X$. In this case, clearly $R_2$ remains on the path, since it did not encounter any exit before $y$, i.e., $\widearc{CM}$ does not have any exit. So  the exit time is $\widearc{CM}+\overline{MX}$, i.e., $ 2y - x$  as shown in Fig.~\ref{fig:diffpointcase1a}.
\begin{figure}[H]
\centering
\includegraphics[height=0.3\linewidth]{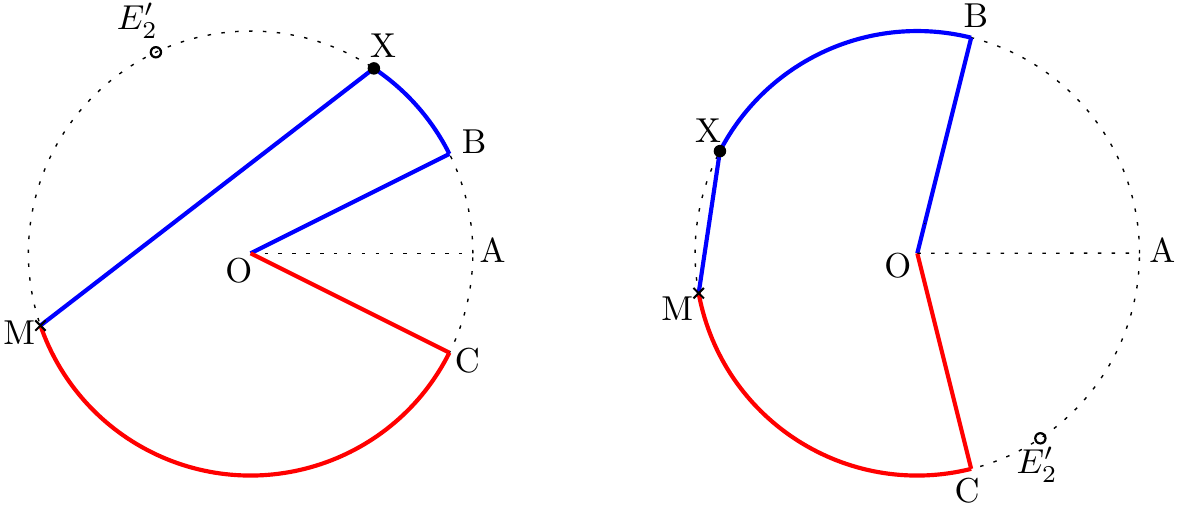}
\caption{$E_2'$ is not in $\widearc{CM}$}\label{fig:diffpointcase1a}
\end{figure}\vspace{-2em}

\textit{Case 1b:} If $E_2'$ lies in $\widearc{CM}$, then $R_1$ moves along $\overline{XE_2'}$ until $N$ such that $x+\overline{XN} = \widearc{CE_2'}+\overline{E_2'N}$ as shown in Fig.~\ref{fig:diffpointcase1b}. If $E_2'$ is the other exit, then the exit time is $x+\overline{XN}+\min(\overline{XN},\overline{E_2'N})$.
\begin{figure}[H]
\centering
\includegraphics[height=0.3\linewidth]{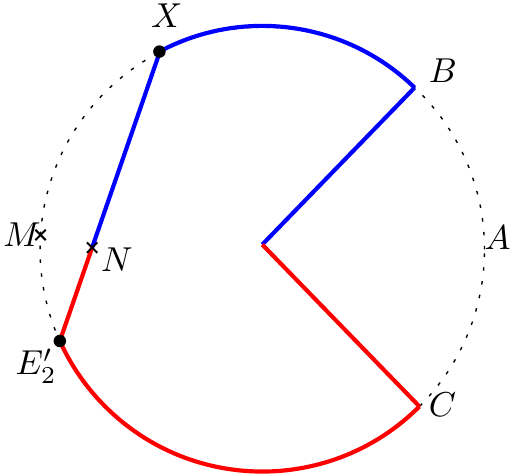}
\caption{$E_2'$ lies in $\widearc{CM}$}\label{fig:diffpointcase1b}
\end{figure}\vspace{-2em}

\textit{Case 1c: }If $E_2'$ is not an exit, then $R_1$ will not encounter $R_2$ at $N$, then it calculates point $P$ such that  $\widearc{CP} = x+\overline{XN}+\overline{NP}$. Then the total time required is $x+\overline{XN}+\overline{NP}+\min(\overline{PX},\overline{PE_1'})$ as shown in Fig.~\ref{fig:diffpointcase1c}.
\begin{figure}[H]
\centering
\includegraphics[height=0.3\linewidth]{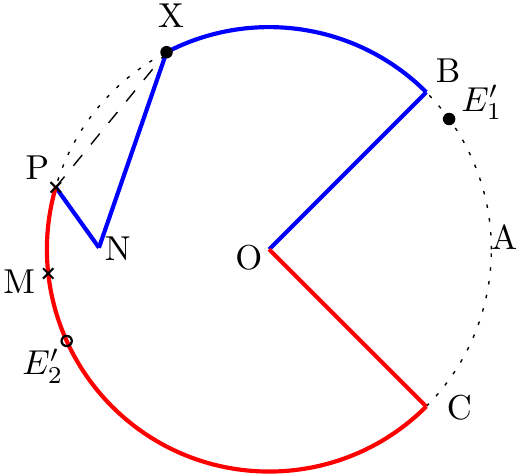}
\caption{$E_2'$ lies in $\widearc{CM}$ and it is not an exit}\label{fig:diffpointcase1c}
\end{figure}\vspace{-2em}

\item[Case 2:] $d \leq x$\\
 $R_1$ knows the probable exit position at $E_2'$ is a real exit.

\textit{Case 2a: } If $\widearc{CE_2'}<d$ as shown in Fig.~\ref{fig:diffpointcase2a}, then $R_1$ moves along $\overline{XE_2'}$ and encounters $R_2$ at $N$. The time for evacuation is $\widearc{XB}+2\overline{XN}$. 
\begin{figure}[H]
\centering
\includegraphics[height=0.3\linewidth]{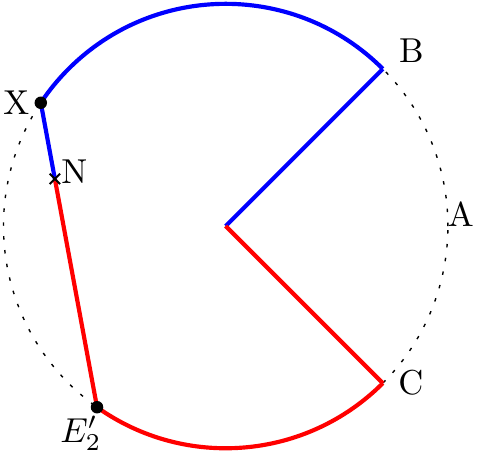}
\caption{$\widearc{CE_2'}<d$}\label{fig:diffpointcase2a}
\end{figure}\vspace{-2em}
\textit{Case 2b: } If $\widearc{CE_2'} \geq d$, then $R_2$ also knows that $X$ is a real exit, so both $R_1$ and $R_2$ will exit via $X$ and $E_2'$ respectively as shown in Fig.~\ref{fig:diffpointcase2b}. Total time required is $\min(x,2\pi -x-2d)$. 
\begin{figure}[H]
\centering
\includegraphics[height=0.3\linewidth]{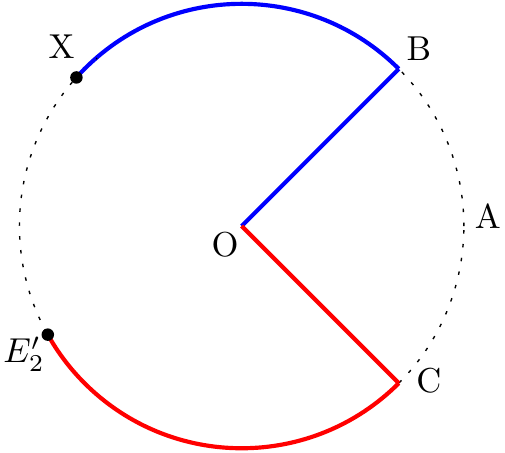}
\caption{$\widearc{BX}>d$ and $\widearc{CE_2'}>d$}\label{fig:diffpointcase2b}
\end{figure}\vspace{-2em}
\textit{Case 2c: } $R_1$ catches $R_2$ before it encounters $E_2'$. Then both robots go to the nearest exit. The time for evacuation is $\widearc{CM} + \min(\overline{MX}, \overline{ME_2'})$ as per Fig.~\ref{fig:diffpointcase2c}.
\begin{figure}[H]
\centering
\includegraphics[height=0.3\linewidth]{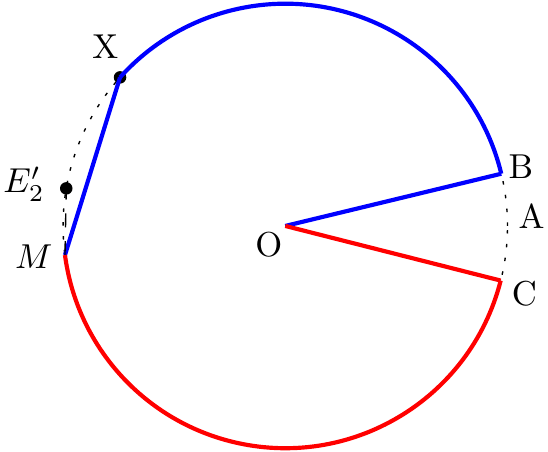}
\caption{$\widearc{CE_2'} > \widearc{CM}$}\label{fig:diffpointcase2c}
\end{figure}\vspace{-2em}
\end{description}

\begin{theorem}
An agreement for evacuation is achieved for $\zeta =d$ in the face-to-face communication model.
\end{theorem}
\begin{proof}
Similar to the proof of Theorem~\ref{thm:f2fzeta0}, the algorithm for $\zeta =d$ in the face-to-face communication model, guarantees that the behavior of the two robots are consistent in all the situations, i.e., either both meet or both exit separately without meeting.
The cases are as following.

\textbf{Case 1:} $ x \leq d$, in this case, $R_1$ behaves based on the location of $E_2'$, where $\widearc{E_2'X} = d$. Clearly, $E_2'$ is the exit in the counter-clockwise direction from $X$. In this case, the exit  in the clockwise direction from $X$ always lies on $\widearc{BC}$. 
If $R_1$ can catch $R_2$ before it reaches $E_2'$, then both robots achieve an agreement at $M$ (ref. Fig.~\ref{fig:diffpointcase1a}). 
Otherwise, $R_1$ tries to meet on the line $XE_2'$. 
If they meet at $N$, then they achieve an agreement (ref. Fig.~\ref{fig:diffpointcase1b} and \ref{fig:diffpointcase2a}). If $R_2$ is not present at $N$, then there is not an exit at $E_2'$. Thence they achieve an agreement at $P$, where $R_1$ catches $R_2$ (ref. Fig.~\ref{fig:diffpointcase1c}).

\textbf{Case 2:} $ x > d$, in this case, $R_1$ knows the position of both exits. If $R_2$ finds an exit and falls in Case 1, then for $R_2$ to achieve an agreement, it needs to know the situation of $R_1$. Hence $R_1$ has to meet $R_2$ at $N$ to achieve  an agreement (ref. Fig.~\ref{fig:diffpointcase1b} and \ref{fig:diffpointcase2a}). If $R_2$ also falls in Case 2, an agreement can be achieved without meeting, since both robots have complete information of the exit locations (ref. Fig.~\ref{fig:diffpointcase2b}). Otherwise, $R_1$ catches $R_2$ before $R_2$ encounters an exit (ref. Fig.~\ref{fig:diffpointcase2c}), and achieve an agreement at $M$.
\end{proof}

\subsection{Labeled Exits}\label{sec:f2flabel}
In this section, we propose a generic algorithm for labeled exits in face-to-face communication model for any $\zeta \leq d$. As the exits are labeled, once a robot encounters an exit, it can determine the location of other exit. Then it checks whether the other robot has encountered an exit or not. Without loss of generality, say $R_1$ finds exit $E_1$ at $X$. It tries to catch $R_2$, if and only if it can catch $R_2$ before $R_2$ encounters an exit. Otherwise, $R_1$ exits at $X$. Since we have considered $R_1$ cannot catch $R_2$ even when it encounters the exit first, this implies that $R_2$ also cannot catch $R_1$ before it encounters an exit. So $R_2$ also exits at its respective exit. This guarantees an agreement for labeled exits. 
Let $\widearc{XB} = x$ and $\widearc{BC} = \zeta$. Let $X'$ be the location of $R_2$ when $R_1$ is at $X$, i.e., $\widearc{CX'} =x$.
There can be four different situations based on the location of $E_1'$ and $E_2'$, where $\widearc{XE_1'} = \widearc{E_2'X} = d$. $R_1$ calculates $M$ where $\widearc{XB} + \overline{XM} = \widearc{CM} = y$. $y$ is the solution to the following equation

\begin{equation}
y = x + 2\sin\left(\frac{x+ y + \zeta}{2}\right)
\end{equation}
 The following cases describe the situations only when $R_1$ encounters the exit before $R_2$. If $R_2$ encounters an exit before $R_1$, then we can consider $R_2$ as $R_1$ and vice versa.
\begin{description}
\item[Case 1:] $E_1'$ does not lie on $\widearc{CM}$, i.e., $x + \zeta + y \leq d$ or $x < d< x + \zeta$. The time for evacuation is $ y + \min(\overline{MX}, \overline{ME_1'})$ as shown in Fig.~\ref{fig:f2flabc1}.
\begin{figure}[H]
\centering
\includegraphics[height=0.3\linewidth]{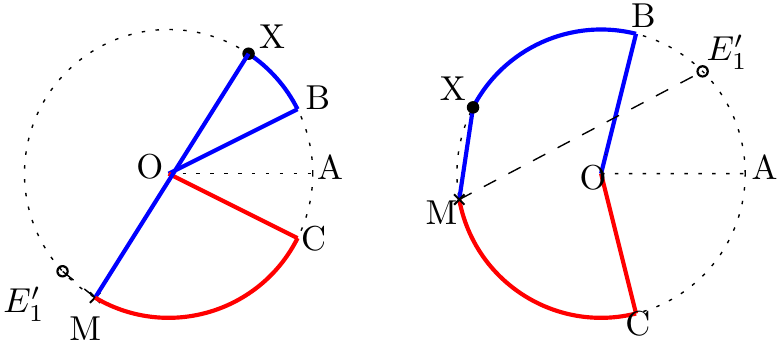}
\caption{$E_1'$ is not in $\widearc{CM}$}\label{fig:f2flabc1}
\end{figure}\vspace{-2em}
\item[Case 2:] $E_1'$ lies on $\widearc{X'M}$, i.e., $2x+\zeta \leq d< x + \zeta+y$. The time for evacuation is $ d- \zeta -x$ as shown in Fig.~\ref{fig:f2flabc2}.
\begin{figure}[H]
\centering
\includegraphics[height=0.3\linewidth]{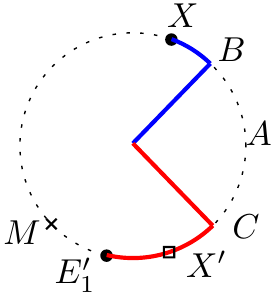}
\caption{$E_1'$ lies on $\widearc{X'M}$}\label{fig:f2flabc2}
\end{figure}\vspace{-2em}
\item[Case 3:] $E_2'$ does not lie on $\widearc{CM}$, i.e., $x + \zeta +y + d< 2\pi $ or $x + d< 2\pi< x+\zeta + d$. The time for evacuation is $ y + \min(\overline{MX}, \overline{ME_2'})$ as shown in Fig.~\ref{fig:f2flabc3}.
\begin{figure}[H]
\centering
\includegraphics[height=0.3\linewidth]{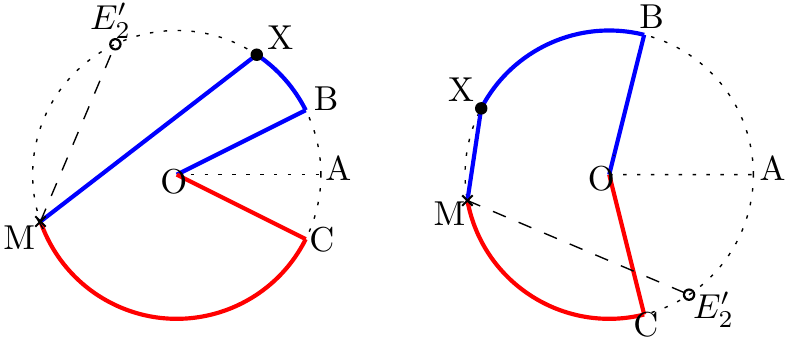}
\caption{$E_2'$ is not in $\widearc{CM}$}\label{fig:f2flabc3}
\end{figure}\vspace{-2em}
\item[Case 4:] $E_2'$ lies on $\widearc{X'M}$, i.e., $2x + \zeta + d< 2\pi \leq x+y+\zeta+d$. The time for evacuation is $2\pi -d-\zeta-x$ as shown in Fig.~\ref{fig:f2flabc4}.
\begin{figure}[H]
\centering
\includegraphics[height=0.3\linewidth]{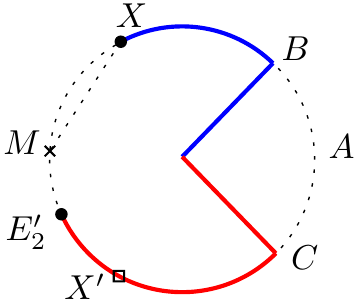}
\caption{$E_2'$ lies on $\widearc{X'M}$}\label{fig:f2flabc4}
\end{figure}\vspace{-2em}
\end{description}

\subsection{Lower Bounds}\label{sec:lowerbound}
In this section, we propose lower bounds for any deterministic algorithm $\mathcal{D}$ for two exits. The general structure of the lower bound follows a common pattern. We consider a polygon where each side of the polygon is of size $v = 2\sin(d/2)$.
Since the distance between two exits is $d$, so if one exit lies on the corner of the polygon, then the other exit also lies on the corner of the polygon. In the following lemmata we show the minimum time required for two robots to exit from the polygon containing two exits.
\begin{lemma}\label{lem:triangleexit}
Consider an equilateral triangle with side length $\sqrt{3}$ with exits situated on two of the vertices of the triangle. In the worst case, it takes at least $\sqrt{3}$ amount of time to exit the disk for both the robots starting from any two arbitrary vertex.
\end{lemma}
\begin{proof}
There are two exits to be placed on three vertices. If one of them does not have the exit, then the other two have the exits. So the robot on the vertex without exit moves to any of the other two vertices with exit in $\sqrt{3}$ amount of time.
\end{proof}
\begin{lemma}\label{lem:traingleexit2}
For $d>2\pi/3$, if there are two points $d$ distance apart and one exit lies on one of the points, then it takes at least $\sin(d)$ amount of time in the worst case.
\end{lemma}
\begin{figure}[H]\centering
\includegraphics[height=0.3\linewidth]{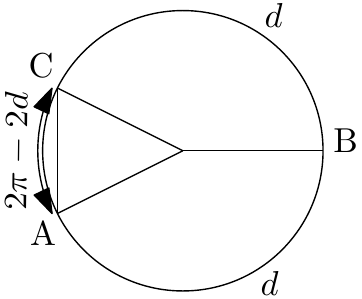}
\caption{The robots start from $A$ and $B$}\label{fig:triangle2}
\end{figure}\vspace{-2em}
\begin{proof}
Consider Fig.~\ref{fig:triangle2}. Suppose the two robots start from the points $A$ and $B$ on the circle. One of those points definitely contains an exit, say that is point $B$. Then the other exit would be at either $A$ or $C$. Then the robot at $A$ has to go to $C$, which is the closest exit to evacuate from the disk. Now $\overline{AC} = 2\sin((2\pi -2d)/2) = 2\sin(d)$.
\end{proof}

\begin{lemma}\label{lem:polygonexit}
Starting from any vertices of the following polygons with corresponding $d$ values, the robots require at least $2$ amount of to exit the disk in the worst case.
\begin{enumerate}
\item Square with side length $\sqrt{2}$, where $d = \pi/2$
\item Pentagon with side length $2 \sin(\pi/5)$, where $d=2\pi/5$
\item Hexagon with side length $1$, where $d=\pi/3$
\end{enumerate}
\end{lemma}
\begin{proof}
\begin{figure}[H]\centering
\includegraphics[height=0.3\linewidth]{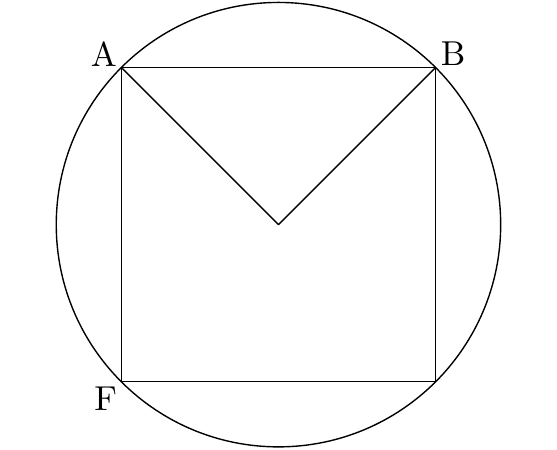}
\includegraphics[height=0.3\linewidth]{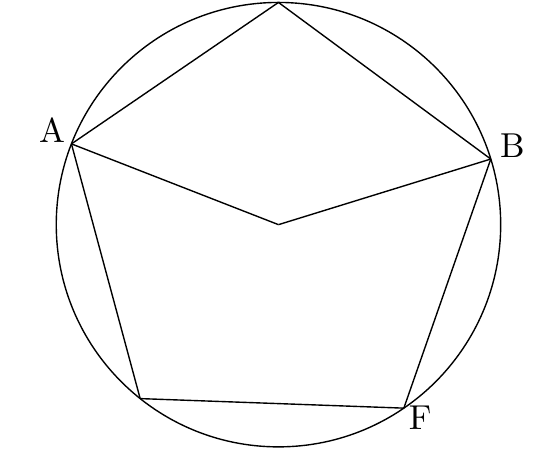}\\
\includegraphics[height=0.3\linewidth]{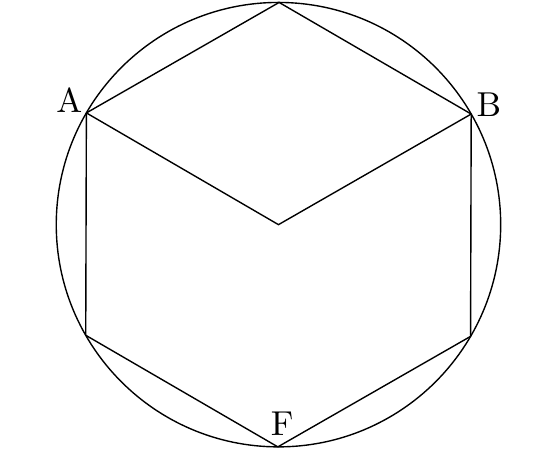}
\caption{Polygons with side length $2\sin(d/2)$, where $d = \pi/2$, $2\pi/5$ and  $\pi/3$}\label{fig:polygon}
\end{figure}

The two robots can start from any of the vertices. These polygons were constructed  such that the sides of polygons corresponds to $d$ in terms of arc length. So if one of the vertices contain an exit, then one of its neighbor also contains a vertex, i.e., two consecutive vertices contain the exits. 
Examine the scenario shown in the Fig.\ref{fig:polygon}. After the two robots meet there can be two cases,
\begin{enumerate}
\item At least one of the two vertices contain an exit.
\item None of the two vertices contains any exit.
\end{enumerate}
If one of the two vertices contain an exit, then they can exit via that exit. If none of the two vertices contain any exit, then the robots can exit via vertex F, marked in the Fig.~\ref{fig:polygon}, which definitely contains an exit.
Since we want to minimize the time required to exit the disk, the meeting point should be equidistant from the three probable exit points, which are the vertices of the polygon. So center of the disk is chosen as the meeting point. Hence the total amount of time required to exit the disk is travelling from vertex to center, i.e., 1  and again center to exit vertex, i.e., 1. Hence 2 unit of time is required in total.
\end{proof}

\begin{theorem}
The lower bound of worst-case evacuation time for two robots with face-to-face communication starting at the center of a unit disk having two exits on its perimeter at a distance $d$ over the arc is the following.
\begin{enumerate}
\item $1+\sin(d)$ for $\pi\geq d > 2\pi/3$
\item $1+\sqrt{3}$ for $2\pi/3 \geq d > \pi/2$
\item $3$ for $0 < d\leq \pi/2$
\end{enumerate}
\end{theorem}
\begin{proof}
From Lemma~\ref{lem:triangleexit}, the robots need at least $\sqrt{3}$ starting from any vertex of a triangle. Traveling to any vertex of the triangle from the center takes 1 amount of time. Hence the total time required to exit the disk is $1+\sqrt{3}$.
Similarly from Lemma~\ref{lem:polygonexit}, we can say that the minimum time needed to exit the disk is $3$, i.e., $2$ for exiting from the polygon, and $1$ to reach the perimeter from the center. 
Now, for the $d$ values which lie in between, the exit time cannot be less than the exit times for regular polygons.
Similarly from Lemma~\ref{lem:traingleexit2}, for $d> 2\pi/3$ it takes $1+\sin(d)$, where 1 being the distance to travel from center to perimeter of disk.
 Hence the theorem follows.
\end{proof}

\section{Simulation Results}\label{sec:simumation}
\subsection{Wireless Communication}
In the simulation part, we have found out the worst case evacuation times for all $d \in [0,\pi]$. We have done the simulation for three values of $\zeta$, namely, $\zeta = 0$, $\zeta = d/2$ and $\zeta = d$. Fig.~\ref{fig:wlcomparision} plots the worst case time for evacuation versus the distance between two exits, $d$. We have done the simulation for $d$ values with a 0.01 step size and the corresponding position of exits with a 0.001 step size. In this we have considered all possible position of exits and found the maximum as the worst case time for evacuation. 
It can be observed from Fig.~\ref{fig:wlcomparision} that the worst case evacuation time is less if the two robots start at a distance $d$ apart initially. Since this is a search problem, reducing the search space can be an effective method to reduce the worst case evacuation time. As we have two exits and we know the distance between them, we can easily remove an arc length equal to $d$ from the perimeter because two exits cannot lie within a $d$ distance  arc.

From Fig.~\ref{fig:wlcomparision}, it is clear that $\zeta = d$ performs better compared to $\zeta=0$ and $\zeta = d/2$. Even, $\zeta =0$ performs better than $\zeta = d/2$ for $d > 1.21$. 
It can be observed from the figure that, it is not strictly monotonous for $\zeta = 0$ or $\zeta = d$. The reason for this is that there is a transition between cases when there is a local minima or local maxima is present.

As per Fig.~\ref{fig:wlcomparision}, for $\zeta = 0$, it is monotonous until $2\pi/3$. For $d \leq 2\pi/3$, it follows the case $d < x$, i.e., one of the probable exits positions is already explored. As this continues the time for evacuation decreases for the same situation as the value of $d$ is increasing. But for $d> 2.09$, it changes to a case where both probable exits are unexplored, hence the time for evacuation increases. Also the plot is curved since the worst case continuously switches between going to the known exit and visiting the probable exits for $d > 2\pi/3$ for $\zeta =0$.

As per Fig.~\ref{fig:wlcomparision}, for $\zeta = d$, observe that there is a local minima  at $ d = 2\pi/3$. The  minima marks the end of the case where $d < x$, that is both exit positions are known. For $0.93 < d \leq 2\pi/3$, both the probable exits are unexplored and so the evacuation time increases and then decreases. 
In the range, where the evacuation time increases until it reaches the local maxima at $d = 1.385$, the unexplored exit $E_1'$ lies on $\widearc{BC}$ very close to $B$.
For $2\pi/3 < d \leq 2.69$, the evacuation time monotonically decreases and the corresponding case has one probable exit explored and the other exit lies in the skipped area. For $d > 2.69$, both the probable exits are unexplored and the worst case situation switches between going towards the known exit and going towards the two probable exits.
\begin{figure}[H]
\centering
\includegraphics[width=0.8\linewidth]{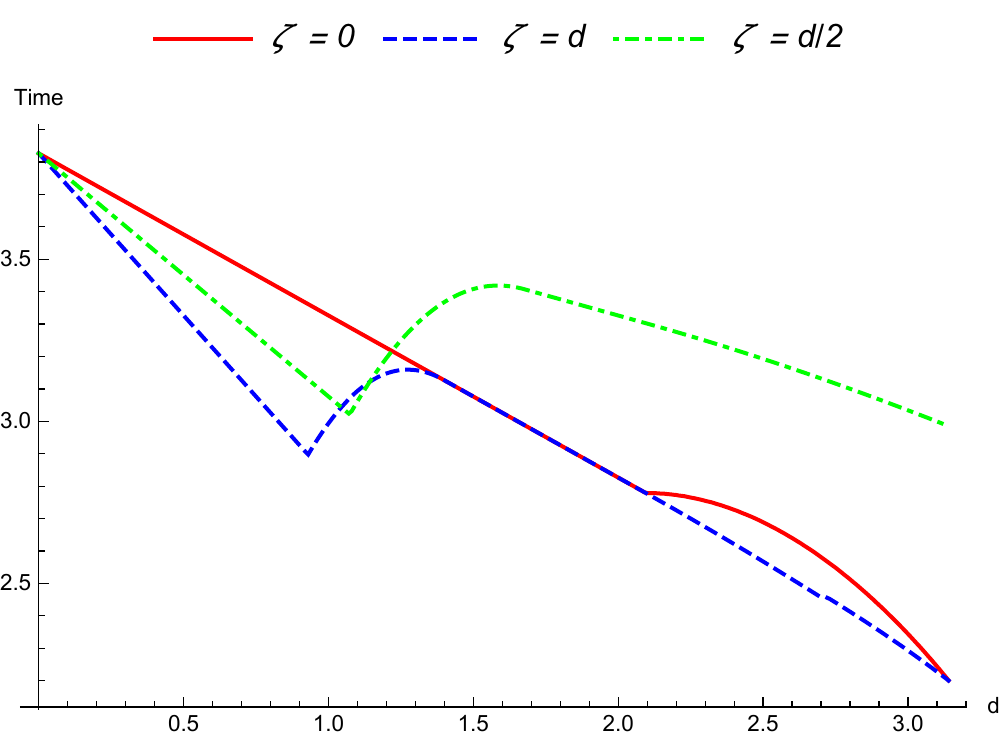}
\caption{Comparision of Algorithms  for different $\zeta$ values for wireless communication with unlabeled exit locations}
\label{fig:wlcomparision}
\end{figure}\vspace{-2em}
\begin{figure}[H]
\centering
\includegraphics[width=0.9\linewidth]{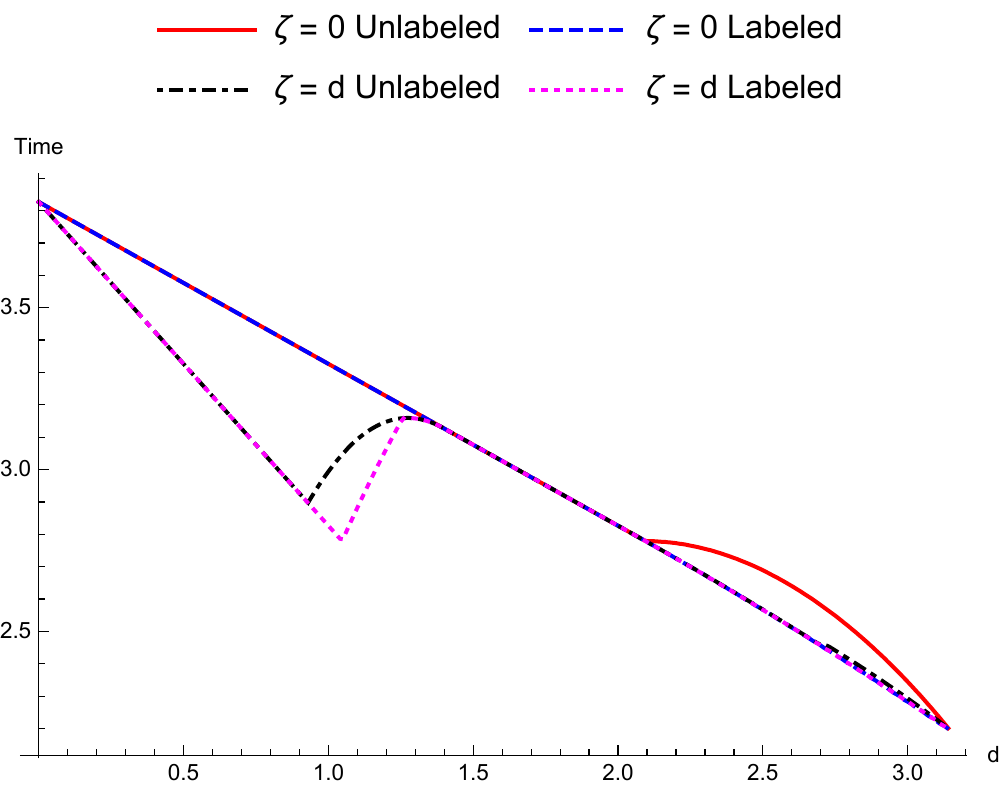}
\caption{Comparision between evacuation time of Unlabeled and Labeled exit locations  for $\zeta =0$}\label{fig:wlcomplower1}
\end{figure}

In Fig.~\ref{fig:wlcomplower1}, we plot the evacuation time for labeled exit location and unlabeled exit locations for $\zeta =0$ and $\zeta =d$. It can be easily observed that the labeled exit time is strictly monotonic for $\zeta =0$. For $\zeta =d$, the evacuation time for labeled exits falls closely with unlabeled exits.

In Fig.~\ref{fig:wlcomplower3}, we plot the evacuation time for labeled exit locations for various values of $\zeta$ such as $\{0.1d, 0.2d,\cdots, d\}$. Then we take the minimum of all these for particular $d$ values and we obtain the $\zeta$ values for which we get the least evacuation time.
Table~\ref{tab:minval} lists the minimum evacuation time for particular $\zeta$ values in wireless communication model. 
\begin{figure}[h]
\centering
\includegraphics[width=0.9\linewidth]{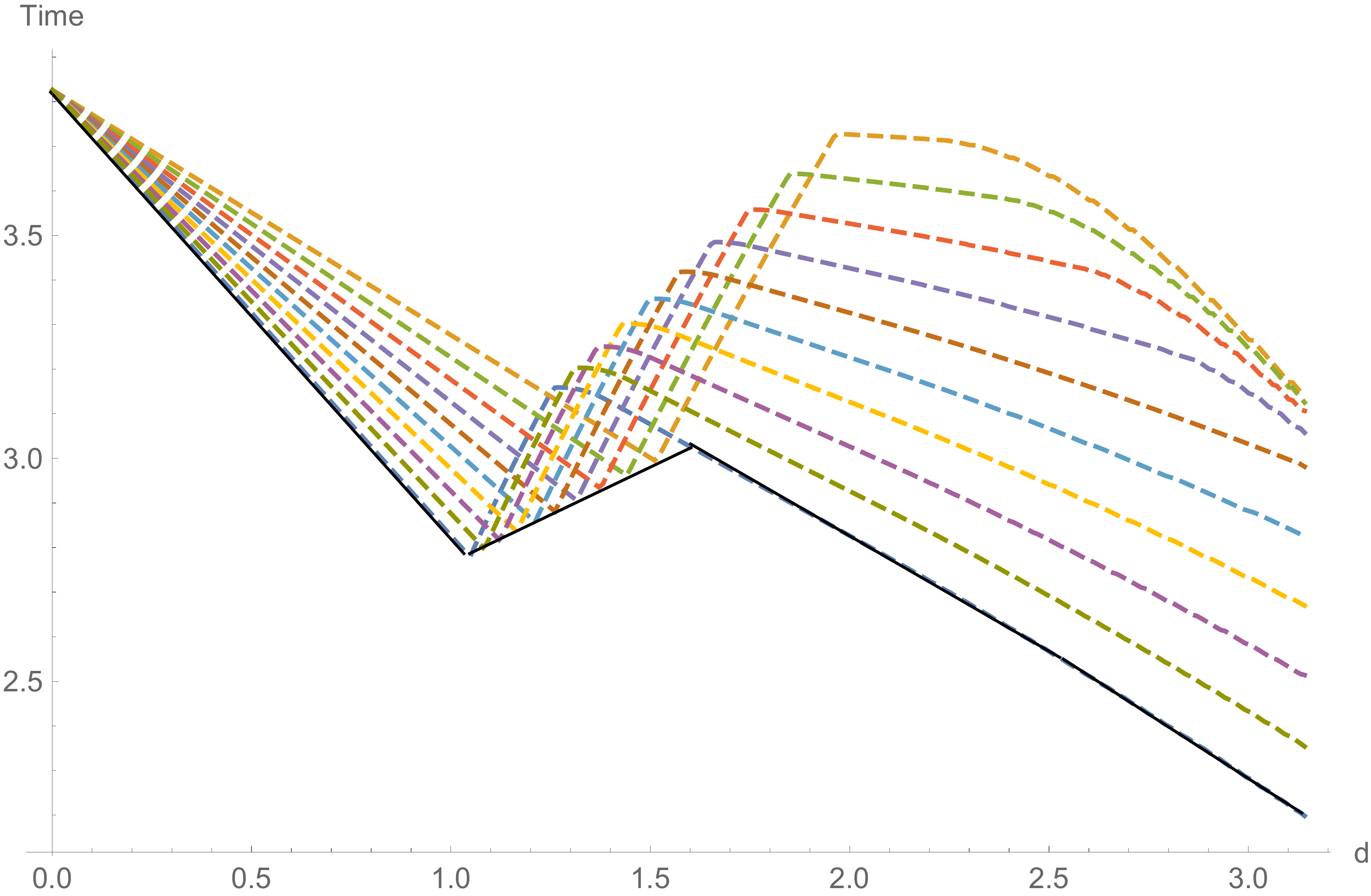}
\caption{Comparision between evacuation time of Labeled exit locations  for different $\zeta$ values in $\{0.1d, 0.2d,\cdots,d\}$ from top to bottom in increasing order}\label{fig:wlcomplower3}
\end{figure}
\begin{table}[H]
\centering
\begin{tabular}{|c|c|c|c|}
\hline
$\zeta$ & \textbf{Minimum Time} & $d$ &\textbf{ Exits} \\\hline
0 & $\pi/4 + \sqrt{2} \approx 2.1996$ & $\pi$ & Unlabeled\\\hline
$d$ & $\pi/4 + \sqrt{2} $ & $\pi$ & Unlabeled\\\hline
$d/2$ & $\pi/2 + \sqrt{2} \approx 2.985$ & $\pi$ & Unlabeled\\\hline
0 & $\pi/4 + \sqrt{2}$ & $\pi$ & Labeled\\\hline
$d$ & $\pi/4 + \sqrt{2}$ & $\pi$ & Labeled\\\hline
$d/2$ & 2.88 & 1.26 & Labeled\\\hline
\end{tabular}
\caption{Minimum evacuation time for different values of $\zeta$ for Labeled and Unlabeled Exits in wireless communication model}\label{tab:minval}
\end{table}

\subsection{Face-to-Face Communication}
In this Section, we compare the algorithms for $\zeta = 0$, in Section~\ref{sec:samepoint} and $\zeta = d$, in Section~\ref{sec:diferentpoint}. The simulation is conducted for $d \in [0,\pi]]$. For the simulation, we have considered discrete points with interval 0.01 for $d$. Then for each $d$, we have taken the location of exits with interval 0.001 in $[0,2\pi]$. For each position of exit, we calculated the time required to exit the disk.
We calculated the solution to the equations using bisection method with error threshold $10^{-6}$.
Then we took the maximum over all those values.
In Fig.~\ref{fig:comparision}, we compare the worst case evacuation time for $\zeta = 0$ and $\zeta = d$ for both labeled and unlabeled exit locations . 
We can observe that, for unlabeled exits, $\zeta = d$ performs better compared to $\zeta = 0$ for all $d$ except $d \in (1.895, 2.005)$ and $d > 2.765$. 

\begin{figure}[H]
\centering
\includegraphics[width=\linewidth]{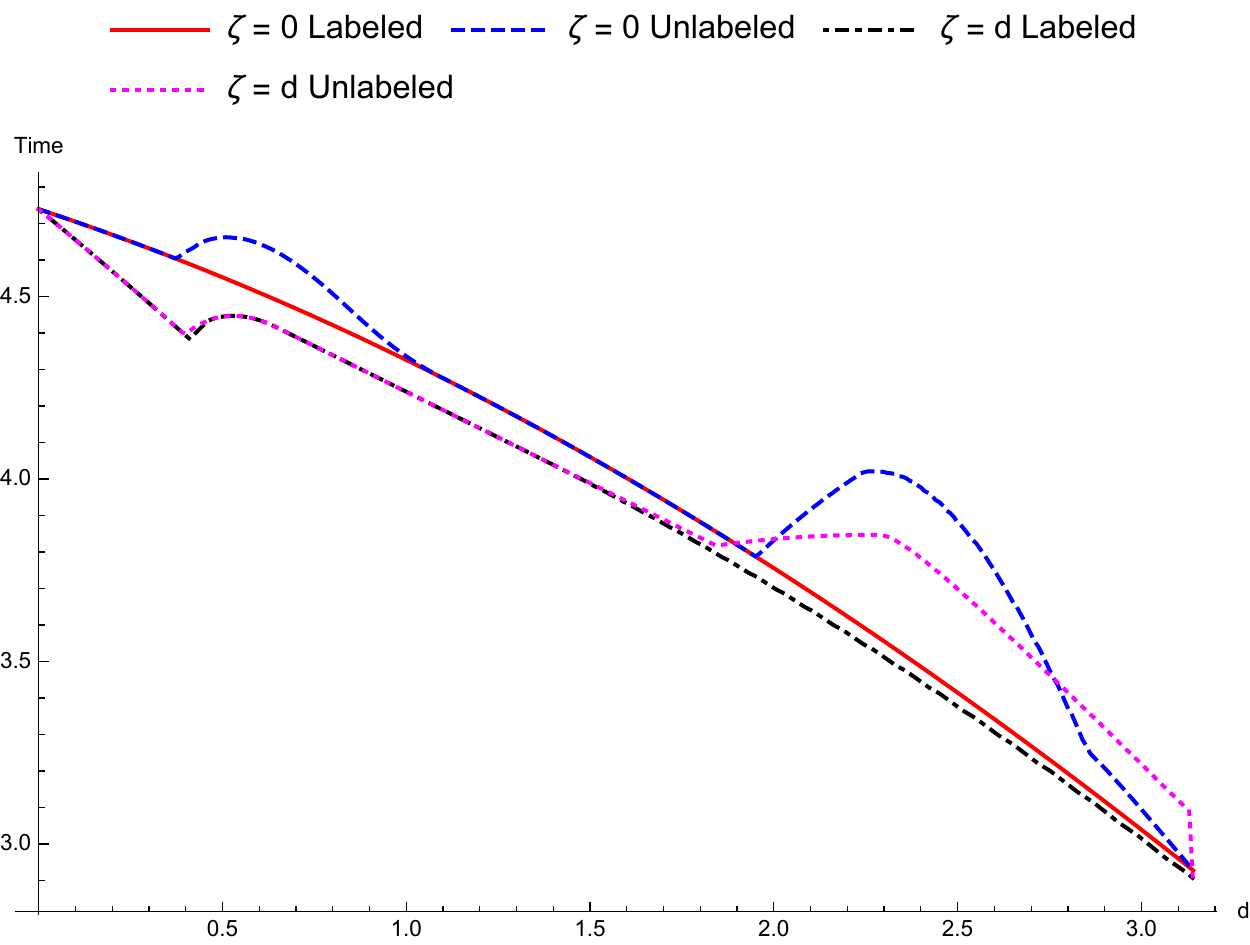}
\caption{Comparision of Algorithms for $\zeta =0$ and $\zeta = d$ in face-to-face communication for Labeled and Unlabeled exits}
\label{fig:comparision}
\end{figure}

It can be observed that the curves are not monotonic for both the algorithms. The change in monotonicity can be attributed to the change in the worst case situations. 

For $\zeta = 0$, the change in monotonic behavior happens as following. At $d = 0.38$, \textit{Case 4a} (ref. Fig.~\ref{fig:samepointcase4a}) to \textit{Case 3a} (ref. Fig.~\ref{fig:samepointcase3cb}). At $d =1.11$ \textit{Case 3a} to \textit{Case 2a} (ref. Fig.~\ref{fig:samepointcase2a}). At $ d = 1.95$ \textit{Case 2a} to \textit{Case 3a} (ref. Fig.~\ref{fig:samepointcase3cb}). At $ 2.52 < d < 2.85$ it switches between \textit{Case 3a} (ref. Fig.~\ref{fig:samepointcase3ca}) and \textit{Case 3a} (ref. Fig.~\ref{fig:samepointcase3cb}). For $ d> 2.85$ it is \textit{Case 1} where the robots move to the two probable exits $E_1'$ and $E_2'$.

For $\zeta = d$, there is a local minima at $d = 0.4$, where the situation changes from \textit{Case 2c} (ref. Fig.~\ref{fig:diffpointcase2c}) to \textit{Case 1a} (ref. Fig.~\ref{fig:diffpointcase1a}). The evacuation time increases, since the robots have knowledge of only one exit instead of two. Then at $d = 1.84$, again the situation changes from \textit{Case 1a} to \textit{Case 1b} (ref. Fig.~\ref{fig:diffpointcase1b}). It continues with \textit{Case 1b} for $1.84 < d < \pi$. At $d = \pi$, there is a sharp drop, since the robots can exactly determine the position of other exit.

As expected, for $\zeta = 0$, evacuation time for labeled exits are strictly monotonic with all the worst case situation occur only in \textit{Case 3} $E_2'$ lies on $\widearc{MX}$. For $\zeta = d$ in case of labeled exits, it is non-monotonic at $d = 0.42$, where it switches from \textit{Case 3} to \textit{Case 1} and it continues with \textit{Case 1} until $d = \pi$.

\section{Conclusion}\label{sec:conclusion}
In this paper, we address the evacuation problem for two robots in a unit disk with two exits located at arbitrary points on the perimeter.
We have considered both unlabeled and labeled exits in the wireless and face-to-face communication models.
As per our findings from the simulation, the evacuation time for labeled exits is always smaller than unlabeled exits for same value of $\zeta$.

The minimum evacuation time for a particular value of $d$ is achieved for a particular value of $\zeta$ in the wireless communication model as per Fig.~\ref{fig:wlcomplower3}. It is still open to find out a function which describes the relation between $\zeta$ and $d$ in the same model. We think that this is very close to the optimal solution for wireless communication model. The previous papers consider only a single exit with $\zeta =0$, and the evacuation time for our algorithm is the same for $d=0$ in the same model.

 The evacuation time for our algorithms in face-to-face communication model are not close to the optimal value. For unlabeled exits, both robots are required to meet for an agreement, even if both have encountered exits, which in turn increases the time traveled. Improvements to our proposed algorithms are possible by using detours similar to~\cite{CzyzowiczGKNOV15}. Although our work for two exits is not comparable, but for $d=0$, i.e., both exits coincide, evacuation time is same as \cite{CzyzowiczGGKMP14}, which is less than \cite{CzyzowiczGKNOV15}.
For the lower bounds, we consider only some specific set of points as probable exit positions, this reduces the search space by a lot. In conclusion, there is scope to improve the bounds for face-to-face communication model.
\bibliographystyle{ACM-Reference-Format}
\bibliography{bib}

\end{document}